\documentclass[twoside]{article}
\usepackage{latexsym,amsmath,amssymb}
\usepackage{sectsty}
\usepackage{hyperref}
  \sectionfont{\large} \subsectionfont{\normalsize}
\usepackage{graphicx}
\usepackage{theorem}
\usepackage[all]{xy}
\newtheorem{theorem}{Theorem}
{\theorembodyfont\rmfamily

}
\newtheorem{remark}{Remark}
\newenvironment{proof}{\begin{paragraph}
          {Proof.}}{\end{paragraph}{(QED)}}

\renewenvironment{abstract}
 {\small\begin{quote}{\textbf{Abstract}}\,\,}{\end{quote}}
\newenvironment{keywords}
 {\small\begin{quote}{\textbf{Keywords}}\,\,}{\end{quote}}
\newenvironment{classification}
 {\small\begin{quote}{\textbf{AMS Subject Classification:}}\,\,}{\end{quote}}
\newenvironment{classification1}
 {\small\begin{quote}{\textbf{JEL:}}\,\,}{\end{quote}}
\date{}

\title{\vspace{-9ex}
{\centering
 \textbf{\large The Use of Numeraires in Multi-dimensional Black-Scholes Partial Differential Equations}}}

 \vspace{-4ex}

\author{\small\textsf{\bfseries 
O, Hyong-chol$^{1,2}$  ~ RO, Yong-hwa$^{1}$ ~ Wan, Ning$^{2}$}\\[-.5ex]
{\footnotesize  ${}^{1}$ Centre of Basic Sciences, \textbf{Kim Il Sung} University, Pyongyang, D. P. R. of Korea}\\[-1ex]
{\footnotesize  ${}^3$Department of Applied Mathematics, Tong-ji University, Shanghai, China}\\[-.5ex]
{\footnotesize Authors: Hyong-chol O, 1964-, Researcher ;  Ning WAN, 1979-, Graduate Student}}

\begin{document}

\maketitle
\thispagestyle{empty}

\vspace{-.6cm}

\begin{abstract}
The  change  of  numeraire  gives  very  important  computational  simplification  in  option  pricing. 
This  technique  reduces  the  number  of  sources  of  risks  that  need  to  be  accounted  for  and  so  it  is  useful  in pricing  complicated  derivatives  that  have  several  sources  of  risks.  
In  this  article,  we  considered  the underlying mathematical theory of numeraire technique in the viewpoint of PED theory and illustrated it with five  concrete  pricing  problems.  
In  the  viewpoint  of  PED  theory,  the  numeraire  technique  is  a  method  of reducing the dimension of status spaces where PDE is defined.
\end{abstract}



\begin{classification1}
\qquad \qquad \qquad \qquad \qquad \qquad G13; G33
\end{classification1}

%
%

\section{Introduction}

The change of numeraire is a very important method of reducing the number of risk sources to consider in option pricing. 
The change of numeraire is well known in theoretical papers but it seems that this method id not often used in application 
and many practitioners don’t know how and when to use it \cite{ben}. Thus in \cite{ben}, they described the essence 
of the method of change of numeraire in the {\it viewpoint of stochastic calculus} and gave 5 typical examples to describe 
the usage of change of numeraire and wrong usage.

The option selecting one asset among two assets with stochastic price is a typical example with several risk sources 
and the change of numeraire has advantage in pricing problems of such options.

The idea to use change of numeraire to simplify option pricing has a nearly the same long history with the famous 
Black-Scholes formula. 
In Black-Scholes formula we can consider US Dollar as a numeraire. 
Merton \cite{mer} used the method of change of numeraire to derive European option pricing formula with zero stochastic yield 
(but he did not use the term “numeraire”). 
Margrabe \cite{mar} first named “{\it numeraire}”. 
In his paper Margrabe acknowledged Steve Ross advising him to use an asset as a numeraire. 
Brenner and Galai \cite{bre} and Fisher \cite{fis} used the method of change of numeraire which we are using now. 
Harrison and Kreps \cite{har} used the price of the security with strictly positive value as a numeraire. 
The asset they used as a numeraire did not have market risk and furthermore had zero interest rate, so the analysis was easier. Geman \cite{gem} and Jamshidian \cite{jam} studied  the mathematical foundation of change of numeraire technique.

This paper can be said to be a PDE edition of \cite{ben}. 
That is, our aim is to describe the essence of the change of numeraire technique in the {\it viewpoint of PDE theory}. 
We will give 4 examples of \cite{ben} and another example of Li and Ren \cite{li} to describe the PDE foundation of the 
usage of change of numeraire:
\begin{itemize}
\item Pricing of employee stock ownership plan (ESOP). 
\item Pricing option with foreign currency strike price.
\item Pricing savings plans that provide a choice of indexing.
\item Pricing convertible bond.
\item Pricing bond which can be converted to stock at maturity.
\end{itemize}

In one word, change of numeraire is a method of reducing space dimension of domain of multi-dimensional PDE 
in viewpoint of PDE theory.

%
%

\section{PDE foundation of Change of Numeraires}

\textbf{Assumptions:} The considered option price depends on $(n+1)$ risk sources.

\begin{enumerate}
\item	$S_0(t), S_1(t), \cdots, S_n(t)$ are the arbitrage-free price processes of tradable assets which don't pay dividend. \label{ee1}
\item	Under the risk neutral measure, $S_0(t), S_1(t), \cdots, S_n(t)$ satisfy 
\begin{equation} \label{eq1}
\frac{dS_i(t)}{S_i(t)} = r(t)dt+\sum_{j=1}^m\sigma_{ij}dW_j(t), \quad i=0, 1, \cdots, n
\end{equation}
Here $r(t)$ is short rate (interest rate), $W_j (j = 1, \cdots, m)$ are one dimensional standard Wiener processes which satisfy 
\begin{align}
& E(dW_i)=0, \label{eq2} \\
& Var(dW_i)=dt, \label{eq3} \\
& Cov(dW_i, dW_j)=0, (i\neq j). \label{eq4}
\end{align} 
\item	 $V(S_0, S_1, \cdots, S_n, t)$ is the arbitrage-free price function of a multi-asset option with maturity $T$.
By Jiang \cite{jia}, $V(S_0, S_1, \cdots, S_n, t)$ satisfies $(n+1)$-dimensional Black-Scholes equation:
\begin{equation} \label{eq5}
\frac{\partial V}{\partial t} + \frac{1}{2}\sum_{i,j=0}^n a_{ij}S_iS_j \frac{\partial^2 V}{\partial S_i \partial S_j} +r(t)\sum_{i=0}^n S_i \frac{\partial V}{\partial S_i}-r(t)V = 0. 
\end{equation} 
Here 
\begin{equation} \label{eq6}
a_{ij} = \sum_{k=0}^n \sigma_{ik} \sigma_{jk}, (i, j = 0, 1, \cdots, n).
\end{equation} 
$A=[a_{ij}]$ is a positive symmetric matrix.
\item The maturity payoff of the option is as follows:
\begin{equation} \label{eq7}
V(S_0, S_1, \cdots, S_n, T) = P(S_0, S_1, \cdots, S_n).
\end{equation} 
\end{enumerate}  
The arbitrage-free price $V(S_0, S_1, \cdots, S_n, t)$ of the option is a solution to the solving problem of PDE \eqref{eq5} 
and \eqref{eq7}. The equation \eqref{eq5} is a parabolic equation with $(n+1)$ spatial variables reflecting $(n+1)$ risk sources.


\begin{theorem} \label{thr1}
Consider the solving problem \eqref{eq5} and \eqref{eq7} in the domain $\{(S_0, S_1, \cdots, S_n, t) \in R^{n+1} \vert t > 0, S_i > 0, i = 0, 1, \cdots, n\}$. If the maturity payoff function $P$ satisfies a homogeneity, that is,
\begin{equation*}
P(aS_0, aS_1, \cdots, aS_n) = aP(S_0, S_1, \cdots, S_n), a>0
\end{equation*}  
then using $S_0$ as a numeraire, that is, change of variable
\begin{equation} \label{eq8}
U=\frac{V}{S_0}, ~~ z_i=\frac{S_i}{S_0}, ~~ i=1, \cdots, n,
\end{equation}  
we can reduce the spatial dimension by one.  Therefore $(n+1)$-dimensional problem \eqref{eq5} and \eqref{eq7} is transformed in to a terminal value problem of an $n$-dimensional Black-Scholes equation. 
\end{theorem}
\begin{proof}
Let denote
\begin{equation*}
F(z_1, \cdots, z_n) \overset{\Delta}{=} P(1, z_1, \cdots, z_n).
\end{equation*}   
Then since by the assumption of homogenety of $P$ we have
\begin{equation*}
P(1, z_1, \cdots, z_n) = \frac{P(S_0, S_1, \cdots, S_n)}{S_0} = \frac{V(S_0, S_1, \cdots, S_n, T)}{S_0},
\end{equation*}    
we can write as follows:
\begin{equation} \label{eq9}
U(z_1, \cdots, z_n, T) = F(z_1, \cdots, z_n).
\end{equation}   
Thus the terminal condition is changed to a $n$-dimensional function. 
In \eqref{eq5} if we change the $S_i$-derivatives of $V$ into $z_i$-derivatives of $U$ using the relation
\begin{equation*}
V(S_0, S_1, \cdots, S_n, t)= U(z_1, \cdots, z_n, t)S_0,
\end{equation*}  
then we have
\begin{equation} \label{eq010}
\frac{\partial U}{\partial t} + \frac{1}{2} \sum_{i,j=1}^n (a_{00}-a_{i0}-a_{0j}+a_{ij})z_i z_j \frac{\partial^2 U}{\partial z_i \partial z_j} = 0.
\end{equation}    
\end{proof}

Now if we can use multi-dimensional Black-Scholes formula (Jiang \cite{jia} (7.3.22)) to get the solution representation to the problem \eqref{eq010} and \eqref{eq9}. 
    
The $n$-dimensional matrix of coefficients of the equation \eqref{eq010}. 
\begin{equation*}
\mathbf{B}_n = [a_{00}-a_{i0}-a_{0j}+a_{ij}]_{i,j=1}^n
\end{equation*}  
is a symmetric positive matrix. 
In fact, the property of symmetry is evident and for every $\boldsymbol{\xi} = (\xi_1, \cdots, \xi_n)^{\bot} \in \mathbf{R}^n$ 
($^{\bot}$ means transposed matrix), 
\begin{align*}
\boldsymbol{\xi}^{\bot} \mathbf{B}_n \boldsymbol{\xi} &= \sum_{i,j=1}^n (a_{00}-a_{i0}-a_{0j}+a_{ij}) \xi_i \xi_j \\
& = a_{00} \left( -\sum_{i=1}^n \xi_i\right)^2 + \sum_{i=1}^n a_{i0}\xi_i \left(-\sum_{j=1}^n \xi_j\right) + \sum_{j=1}^n a_{0j} \left(-\sum_{i=1}^n \xi_i\right) \xi_j + \sum_{i,j=1}^n a_{ij}\xi_i \xi_j.
\end{align*}   
Therefore if we let  
\begin{equation*}
\boldsymbol{\eta} = \left( -\sum_{j=1}^n \xi_i, \xi_1, \cdots, \xi_n \right)^{\bot} \in \mathbf{R}^{n+1},
\end{equation*}  
then by positiveness of $\mathbf{A}$, we have
\begin{equation*}
\boldsymbol{\xi}^{\bot}\mathbf{B}_n \boldsymbol{\xi} = \boldsymbol{\eta}^{\bot} \mathbf{A} \boldsymbol{\eta} \geq 0.
\end{equation*}
Therefore \eqref{eq010} is $n$-dimensional Black-Scholes equation and by multi-dimensional Black-Scholes formula, we have
\begin{equation*}
U(Z, t) = \left[ \frac{1}{2\pi(T-t)}\right]^{\frac{n}{2}} \left\vert \det \mathbf{B}_n \right\vert^{-\frac{1}{2}} \int_{0}^{\infty} \cdots \int_{0}^{\infty} \frac{F(y_1, \cdots, y_n)}{y_1, \cdots, y_n} \exp \left[-\frac{\vec{\alpha}^{\bot}\mathbf{B}_n^{-1}\vec{\alpha}}{2(T-t)}\right] dy_1 \cdots dy_n,
\end{equation*}  
where 
\begin{equation*}
Z=(z_1, \cdots, z_n)^{\bot}, ~~ \vec{\alpha} = (\alpha_1, \cdots, \alpha_n)^{\bot}
\end{equation*} 
and
\begin{equation*}
\alpha_i = \ln \frac{z_i}{y_i}-\frac{a_{00}-2a_{i0}+a_{ii}}{2}(T-t), ~~ i=1, \cdots, n.
\end{equation*}  
Returning to the original variables $S_0, S_1, \cdots, S_n$, then we have the solution of \eqref{eq5} and \eqref{eq7}:
\begin{align}
& V(S_0, S_1, \cdots, S_n, t) = S_0U(Z, t) = S_0 \left[ \frac{1}{2\pi (T-t)}\right]^{\frac{n}{2}}\left\vert \det \mathbf{B}_n \right\vert^{-\frac{1}{2}}  \nonumber \\
& \qquad \qquad \times \int_{0}^{\infty} \cdots \int_{0}^{\infty} \frac{P(1, y_1, \cdots, y_n)}{y_1, \cdots, y_n} \exp \left[-\frac{\vec{\alpha}^{\bot}\mathbf{B}_n^{-1}\vec{\alpha}}{2(T-t)}\right] dy_1 \cdots dy_n, \label{eq011}\\
& \alpha_i = \ln \frac{S_i}{y_i S_0}-\frac{a_{00}-2a_{i0}+a_{ii}}{2}(T-t), ~~ i=1, \cdots, n. \label{eq012}
\end{align}    
Thus we proved the following theorem:


\begin{theorem} \label{thr2}
The $n$-dimensional representation of $(n+1)$ dimensional problem \eqref{eq5} and \eqref{eq7} is provided by \eqref{eq011} and \eqref{eq012}.
\end{theorem}
\begin{remark}
No dividend condition \ref{ee1} and no correlation assumption \eqref{eq4} are not essential and only for simplification.
\end{remark}

%
%

\section{Examples}

%
%

\subsection{Employee Stock Ownership Plan (ESOP)}

\textbf{Problem:} ESOP is a contract, the holder of which has the right to buy a stock for the following strike price. 
For example, let the maturity be a year $(T_1)$. 
Then the strike price is 
min (the stock price at the time $(T_0)$ after 6 months, the stock price after a year $(T_1)$) $\times 85\% (\beta)$.
Then what is the fair price of ESOP? \\

\noindent \textbf{Mathematical model:}

Let $S(t)$ be the stock price, $T_1$ the maturity, $T_0 < T_1,  0 <\beta < 1$ the discount rate and $V(S, t)$ the price of ESOP.  Then the maturity payoff is as follows:
\begin{equation} \label{eq013}
V(S, T_1) = [S(T_1)-\beta \min (S(T_1), S(T_0))]^{+} = S(T_1)-\beta \min (S(T_1), S(T_0)).
\end{equation} 
 
\noindent \textbf{Assumptions:} 
\begin{enumerate}
\item	Under the risk neutral measure, the stock price satisfies
\begin{equation}\label{eq014}
dS(t) = rS(t)dt + \sigma S(t)dW(t).
\end{equation}     
\item The short rate $r$ is deterministic constant.
\end{enumerate}

In fact the problem itself is relatively simple and so without using change of numeraire we can easily get solution formula. The risk neutral price of  the contract satisfies standard Black-Scholes equation: 
\begin{equation} \label{eq015}
\frac{\partial V}{\partial t} + \frac{1}{2}\sigma^2 S^2 \frac{\partial^2 V}{\partial S^2} +rS \frac{\partial V}{\partial S}-rV = 0.
\end{equation}   
Therefore the price function of ESOP satisfies the terminal value problem \eqref{eq015} and \eqref{eq013}.  
The maturity payoff \eqref{eq013} is rewritten as follows:
\begin{equation*}
V(S, T_1) = (1-\beta)S(T_1)+\beta \max (S(T_1)-S(T_0), 0).
\end{equation*}  
In the interval $(T_0, T_1]$ the price $S(T_0)$ is already known and thus the contract can be seen as a portfolio consisting 
of $1-\beta$ shares of stock and $\beta$ sheets of option with strike price $S(T_0)$. 
So the price of  the contract at time $t (T_0 \leq t < T_1)$ is given by
\begin{align*}
& V(S, t) = (1-\beta)S(t)+\beta[S(t)N(d_1)-S(T_0)e^{-r(T_1-t)}N(d_2)], ~~ T_0 \leq t < T_1 \\
& d_1 = \frac{\ln \left( \frac{S(t)}{S(T_0)}\right) + \left(r+\frac{1}{2}\sigma^2\right)(T_1-t)}{\sigma \sqrt{T_1-t}}, ~~ d_2 = d_1-\sigma \sqrt{T_1-t}. 
\end{align*}  
In particular at time $T_0$,
\begin{equation*}
V(S, T_0) = S(T_0)\left[ 1-\beta+\beta N(d_{1, 0})-\beta e^{-r(T_1-T_0)}N(d_{2, 0})\right], 
\end{equation*}
\begin{equation} \label{eq016}
d_{1, 0} = \left( \frac{r}{\sigma} + \frac{\sigma}{2} \right) \sqrt{T_1-T_0}, ~~ d_{2, 0} = d_{1, 0} - \sigma \sqrt{T_1-T_0}.
\end{equation}   
So its price at time $T_0$ is just the same with the price of deterministic shares of stock, thus the price of the contract at time 
$t (0 \leq t < T_0)$ is given by
\begin{equation} \label{eq017}
V(S, t) = S(t)\left[ 1-\beta+\beta N(d_{1, 0})-\beta e^{-r(T_1-T_0)}N(d_{2, 0})\right], ~~ 0 \leq t \leq T_0,
\end{equation}   
where $d_{1, 0}$ and $d_{2, 0}$ are as in \eqref{eq016}. \\
 
\noindent \textbf{Use of change of numeraire:} In \eqref{eq013} it is not so natural to see $S(T_0)$ as the time $t$-price of a tradable asset. To overcome this difficulty, in \cite{ben} a new asset $S_0$ defined as follows:  
\begin{equation*}
S_0(t) = \left\{ 
\begin{array}{ll}
S(t), & 0 \leq t \leq T_0, \\
S(T_0)e^{r(t-T_0)}, & T_0 \leq t \leq T_1. 
\end{array}
\right.
\end{equation*}
 
Note that $S_0$ is the time $t$-price of the portfolio, the holder of which buys a share of stock at time $t = 0$, 
keeps it until $t = T_0$, then sells it at time $t = T_0$ and saves the cash of the quantity of $S(T_0)$ into bank. 
Then 
\begin{equation*}
S_0(T_1) = S(T_0)e^{r(T_1-T_0)},
\end{equation*} 
and if we let 
\begin{equation} \label{eq018}
K=e^{-r(T_1-T_0)},
\end{equation}  
then the maturity payoff \eqref{eq013} is rewritten as follows: 
\begin{equation} \label{eq019}
V(S, S_0, T_1) = S(T_1) - \beta \min (S(T_1), K \cdot S_0(T_1)).
\end{equation}  
In \eqref{eq019}, $S_0(T_1)$ can be seen as the price at time $T_1$ of a tradable asset $S_0$, and from the definition of $S_0$ and the assumption \eqref{eq014}, we can know that under the risk neutral measure, $S_0$ satisfies
\begin{equation} \label{eq020}
dS_0(t) = rS_0(t)dt + \sigma_0 S_0(t)dW(t),
\end{equation}   
where
\begin{equation*}
\sigma_0(t) = \left\{ 
\begin{array}{ll}
\sigma, & 0 \leq t \leq T_0, \\
0, & T_0 \leq t \leq T_1. 
\end{array}
\right.
\end{equation*} 

Using no arbitrage technique, we can establish the equation that the risk neutral price function $V(S, S_0, t)$ of our contract satisfies. By $\Delta$-hedging construct a portfolio $\Pi$: 
\begin{equation*}
\Pi = V-\Delta S - \Delta_0 S_0.
\end{equation*}  
Select $\Delta$ and $\Delta_0$ so that the portfolio $\Pi$ becomes risk free in the interval $(t, t + dt)$, that is,
\begin{equation*}
d \Pi = r \Pi dt
\end{equation*}   
and thus we have
\begin{equation} \label{eq021}
dV-\Delta dS-\Delta_0 dS_0 = r(V-\Delta S - \Delta_0 S_0)dt.
\end{equation}    
By 2-dimensional It\^o formula, we get
\begin{equation} \label{eq022}
dV = \left\{ \frac{\partial V}{\partial t} + \frac{1}{2} \left[ \sigma^2 S^2 \frac{\partial^2 V}{\partial S^2} +2\sigma \sigma_0(t) S S_0 \frac{\partial^2 V}{\partial S \partial S_2} +\sigma_0^2 (t) S_0^2 \frac{\partial^2 V}{\partial S_0^2} \right] \right\}dt + \frac{\partial V}{\partial S}dS + \frac{\partial V}{\partial S_0}dS_0.
\end{equation} 
Thus we select
\begin{equation} \label{eq023}
\Delta = \frac{\partial V}{\partial S}, ~~  \Delta_0 = \frac{\partial V}{\partial S_0},
\end{equation}  
and we substitute \eqref{eq022} and \eqref{eq023} into \eqref{eq021} to get 
\begin{equation} \label{eq024}
\frac{\partial V}{\partial t} + \frac{1}{2} \left[ \sigma^2 S^2 \frac{\partial^2 V}{\partial S^2} +2\sigma \sigma_0(t) S S_0 \frac{\partial^2 V}{\partial S \partial S_2} +\sigma_0^2 (t) S_0^2 \frac{\partial^2 V}{\partial S_0^2} \right] + r\frac{\partial V}{\partial S}S + r \frac{\partial V}{\partial S_0}S_0-rV = 0.
\end{equation} 
Therefore the two factor model of ESOP price is just the terminal value problem \eqref{eq024} and \eqref{eq019}. 
Except for the fact that the volatility $\sigma_0$ depends in time, the problem \eqref{eq024} and \eqref{eq019} 
satisfies the all conditions of Theorem \ref{thr1}. In this case the change of numeraire 
 \begin{equation} \label{eq025}
U=\frac{V}{S_0}, ~~ z=\frac{S}{S_0}
\end{equation} 
does work well, that is, using \eqref{eq025} the problem \eqref{eq024} and \eqref{eq019} is transformed into the following terminal problem of one dimensional Black-Scholes equation: 
\begin{equation}\label{eq026}
\left\{ 
\begin{array}{l}
\frac{\partial U}{\partial t} + \frac{1}{2}(\sigma-\sigma_0(t))^2 z^2 \frac{\partial^2 U}{\partial z^2} = 0, \\
U(z, T_1) = (1-\beta_1) z + \beta \max (z-K, 0).
\end{array}
\right.
\end{equation}  
Solving \eqref{eq026} and returning to the original variables, then we get the formula \eqref{eq017}.
\begin{remark}
The formula \eqref{eq017} is the same with that of \cite{ben}.
\end{remark}
   
%
%

\subsection{Pricing Options with Foreign Currency Strike Price}

Some option's strike price is related to foreign currency. 
For example there is an option whose listed underlying stock price is UK pounds while the strike price is US dollars. 
The aim of designing such options is to stimulate financial managers so as to maximize US dollar price of the stock.\\

\noindent \textbf{Problem:} Assume that an option's underlying stock is traded using UK pounds while the strike price is paid using US dollars: 
1) the stock price $S(0)$ is the UK pound price of the strike price (because of tax and etc., most executing stock options are always at the money at the initial time);  
2) the US dollar price of $S(0)$ is just the constant strike price;  
3) the holder of the option can buy the underlying stock for the fixed US dollar strike price.

Since the stock is traded using UK pounds and the fixed US dollar strike price corresponds to stochastic UK pounds strike price, the pricing of this option is {\it not a standard} pricing problem. Using change of numeraire, we can simplify this pricing problem.\\

\noindent \textbf{Mathematical model:}
  
Let $S(t)$ denote the stock price (UK pounds), $r_p$ UK (risk free) interest rate, $r_d$ US (risk free) interest rate, 
$X(t)$ US dollar/UK pound exchange rate ($Y(t) = X(t)^{-1}$: UK pound/US dollar exchange rate), $K_d$  the fixed US dollar strike price and $K_p(t)$ the UK pound price of the strike price. \\

\noindent \textbf{Assumptions:} 
\begin{enumerate}
\item	The stock price satisfies the following geometrical Brownian motion
\begin{equation*} 
dS(t) = \alpha_S S(t)dt + \sigma_S S(t)dW^S(t)
\end{equation*} 
under a natural measure.
\item Risk free rates $r_p$ and $r_d$ are deterministic constant.
\item US dollar/UK pound exchange rate $X(t)$ satisfies the model of Garman and Kohlhagen \cite{gar}
\begin{equation*} 
dX(t) = \alpha_X X(t)dt + \sigma_X X(t)dW^X(t) (\textnormal{under a natural measure}).
\end{equation*}
\end{enumerate} 
By It\^o formula, UK pound/US dollar exchange rate $Y(t)$ satisfies 
\begin{equation*} 
dY(t) = \alpha_Y Y(t)dt + \sigma_Y Y(t)dW^Y(t) (\textnormal{under a natural measure}).
\end{equation*} 
Here
\begin{equation*}
\sigma_Y = \sigma_X, ~~ W^Y = -W^X.
\end{equation*}  
$W^S(t)$ and $W^X(t)$ are vector Wiener processes with the following relations:
\begin{equation}\label{eq027}
\begin{array}{l}
dW^S(t) \cdot dW^X(t) = \rho dt, \\
dW^S(t) \cdot dW^Y(t) = - \rho dt, ~~ \vert \rho \vert <1.
\end{array}
\end{equation}   
As for {\bf the strike price} we have
\begin{align*}
& K_p(0) = S(0), \\
& K_d =K_p(0) \cdot X(0) = S(0) \cdot X(0) \equiv \textnormal{constant}.
\end{align*}    
The US dollar strike price is a constant but the exchange rate is stochastic, and so the UK pound price of the strike price is a stochastic variable:
\begin{equation*}
K_p(t) = K_d \cdot Y(t) = S(0) \cdot X(0) \cdot X(t)^{-1}.
\end{equation*}  
In particular
\begin{equation}\label{eq028}
K_p(T) = S(0) \cdot X(0) \cdot X(T)^{-1} = S(0) \cdot Y(T) \cdot Y(0)^{-1}.
\end{equation}   
As for {\bf the maturity payoff},  the maturity payoff in US dollar is
\begin{equation*}
F_d = \max (S(T) \cdot X(t)-K_d, 0).
\end{equation*} 
By \eqref{eq028} the maturity payoff in UK pound is
\begin{equation*}
F_p = \max (S(T)-K_p(T), 0) = \max (S(T)-S(0) \cdot Y(0)^{-1} \cdot Y(T), 0).
\end{equation*} 

\begin{remark}
There are two natural methods in pricing this option: US dollar pricing and UK pound pricing. 
In UK pound pricing the use of theorem 1 can transform 2-dimensional problem into one-dimensional problem while in US 
dollar pricing the maturity payoff function does not satisfy the condition of theorem 1 but another change of variables are valid.
 \end{remark} 
    
\noindent \textbf{US dollar pricing}

Let $V_d = V(S, X, t)$ be the US dollar price of our option. Using $\Delta$-hedging, construct a portfolio
\begin{equation*}
\Pi = V-\Delta_1 SX- \Delta_2X.
\end{equation*}  
(This portfolio consists of a sheet of option, $\Delta_1$ shares of stock and $\Delta_2$ UK pounds and its price is calculated in US dollar.) We select $\Delta_1$ and $\Delta_2$ so that $\Pi$ becomes risk free in $(t, t + dt)$, that is,
\begin{equation*}
d\Pi = r_d \Pi dt
\end{equation*}   
which is equivalent to
\begin{equation}\label{eq029}
d\Pi = dV - \Delta_1 d(SX) - \Delta_2dX - \Delta_2 r_p dt \cdot X = r_d (V-\Delta_1 SX - \Delta_2 X)dt.
\end{equation}   
By It\^o formula, we have
\begin{align*}
& dV = \left\{ \frac{\partial V}{\partial t} + \frac{1}{2} \left[ \sigma_S^2 S^2 \frac{\partial^2 V}{\partial S^2} +2\rho \sigma_S \sigma_X SX \frac{\partial^2 V}{\partial S \partial X} + \sigma_X^2 X^2 \frac{\partial^2 V}{\partial X^2}\right]\right\}dt \\
& \qquad \qquad + \frac{\partial V}{\partial S}dS + \frac{\partial V}{\partial X}dX, \\
& d(SX) = SdX + XdS + \sigma_S \sigma_X SX \rho dt.
\end{align*} 
Substituting the above equalities into \eqref{eq029} we have 
\begin{align*}
& \left\{ \frac{\partial V}{\partial t} + \frac{1}{2} \left[ \sigma_S^2 S^2 \frac{\partial^2 V}{\partial S^2} +2\rho \sigma_S \sigma_X SX \frac{\partial^2 V}{\partial S \partial X} + \sigma_X^2 X^2 \frac{\partial^2 V}{\partial X^2}\right]\right\}dt + \left( \frac{\partial V}{\partial S}- \Delta_1 X \right) dS \\
& \qquad \qquad + \left( \frac{\partial V}{\partial X}-\Delta_1S-\Delta_2 \right)dX - \Delta_1 \sigma_S \sigma_X SX \rho dt - \Delta_2 Xr_p dt \\
&  \qquad = r_d (V-\Delta_1 SX - \Delta_2 X)dt.
\end{align*}  
Thus if we select
\begin{equation}
\Delta_1 = \frac{1}{X} \frac{\partial V}{\partial S}, ~~ \Delta_2 = \frac{\partial V}{\partial X} - \Delta_1 S = \frac{\partial V}{\partial X} - \frac{S}{X} \frac{\partial V}{\partial S},
\end{equation}    
then we have the PDE model of US dollar price of the option:
\begin{align}
& \frac{\partial V}{\partial t} + \frac{1}{2} \left[ \sigma_S^2 S^2 \frac{\partial^2 V}{\partial S^2} + 2\rho \sigma_S \sigma_X SX \frac{\partial^2 V}{\partial S \partial X} + \sigma_X^2 X^2 \frac{\partial^2 V}{\partial X^2} \right] \nonumber \\
& \qquad +(r_p-\rho \sigma_S \sigma_X)S \frac{\partial V}{\partial S} + (r_d-r_p)X \frac{\partial V}{\partial X} - r_dV = 0, \label{eq031}\\
& V(S, X, T) = \max (S \cdot X - K_d, 0). \label{eq032}
\end{align}    
The maturity payoff function \eqref{eq032} does not satisfy the homogeneity condition but the use of the change of variable 
\begin{equation} \label{eq033}
z=SX
\end{equation} 
(the essence of \eqref{eq033} is to change UK pound price to US dollar price.) transforms the problem \eqref{eq031} and \eqref{eq032} into one dimensional problem: 
\begin{equation}\label{eq034}
\left\{
\begin{array}{l}
\frac{\partial V}{\partial t} + \frac{1}{2} (\sigma_S^2 + 2\rho \sigma_S \sigma_X + \sigma_X^2)z^2 \frac{\partial^2 V}{\partial z^2} + r_d z \frac{\partial V}{\partial z} -r_d V = 0, \\
V(z, T) = \max (z-K_d, 0).
\end{array}
\right.
\end{equation}   
This is an ordinary call option (in US dollar) problem and thus by standard Black-Schole formula 
we have
\begin{align*}
& V(z, t) = zN(d_1) - K_de^{-r_d(T-t)}N(d_2), \\
& d_1 = \frac{\ln \frac{z}{K_d}+\left( r_d + \frac{1}{2} \sigma_{S, X}^2 \right)(T-t)}{\sigma_{S, X}\sqrt{T-t}}, ~~ d_2 = d_1-\sigma_{S, X} \sqrt{T-t}.
\end{align*}
Here
\begin{equation} \label{eq035}
\sigma_{S, X}^2 = \sigma_S^2 + 2\rho\sigma_S \sigma_X + \sigma_X^2.
\end{equation}  
Returning to the original variables, we get the {\bf US dollar price} of the option:
\begin{equation} \label{eq036}
V_d(S, X, t) = SXN(d_1) - K_d e^{-r_d(T-t)} N(d_2),
\end{equation}   
\begin{equation*} 
d_1 = (\sigma_{S, X}\sqrt{T-t})^{-1} \left[ \ln(SX/K_d) + (r_d+\sigma_{S, X}^2/2)(T-t)\right], ~ d_2 = d_1-\sigma_{S, X}\sqrt{T-t}.
\end{equation*}  
From the fact that $V_p = V_d X^{-1}$ and $K_d = S(0)X(0)$, we get the {\bf UK pound price} of the option:
\begin{equation} \label{eq037}
V_p(S, X, t) = SN(d_1) - S_0X_0X^{-1}e^{-r_d(T-t)}N(d_2),
\end{equation}   
\begin{align*}
& d_1 = \frac{\ln \frac{SX}{S_0X_0} + \left( r_d+\frac{1}{2}\sigma_{S, X}^2 \right)(T-t)}{\sigma_{S, X}\sqrt{T-t}}, \\
& d_2 = d_1 - \sigma_{S, X} \sqrt{T-t}.
\end{align*}

In {\bf Direct Calculation of UK pound price} the change of numeraire technique works well. Remind that
\begin{equation*} 
F_p = \max(S(T)-S(0) \cdot Y(0)^{-1} \cdot Y(T), 0).
\end{equation*}  
Let $V_p(S, Y, t)$ be the UK pound price function of the option. By $\Delta$-hedging, construct a portfolio 
\begin{equation*} 
\Pi = V - \Delta_1S - \Delta_2 Y.
\end{equation*}  
(This portfolio consists of a sheet of option, $\Delta_1$ share of stock and $\Delta_2$ US dollar and its price is calculated in UK pound.) We select $\Delta_1$ and $\Delta_2$ so that $\Pi$ becomes risk free in $(t, t + dt)$, that is,
\begin{equation*} 
d\Pi = r_p \Pi dt
\end{equation*}   
which is equivalent to
\begin{equation} \label{eq038}
dV-\Delta_1dS-\Delta_2dY-\Delta_2r_d dt \cdot Y = r_p(V-\Delta_1 S-\Delta_2 Y)dt.
\end{equation}    
By It\^o formula, we have
\begin{equation} \label{eq039}
dV = \left\{ \frac{\partial V}{\partial t}+\frac{1}{2}\left[ \sigma_S^2 S^2 \frac{\partial^2 V}{\partial S^2}-2\rho\sigma_S \sigma_Y SY \frac{\partial^2 V}{\partial S \partial Y} + \sigma_Y^2 Y^2 \frac{\partial^2 V}{\partial Y^2} \right] \right\}dt + \frac{\partial V}{\partial S}dS + \frac{\partial V}{\partial Y}dY.
\end{equation}    
Thus if we select
\begin{equation} \label{eq040}
\Delta_1 = \frac{\partial V}{\partial S}, ~~ \Delta_2 = \frac{\partial V}{\partial Y},
\end{equation}   
and substitute \eqref{eq039} and \eqref{eq040} into \eqref{eq038}, then we have the PDE model of UK pound pricing:
\begin{align}
& \frac{\partial V}{\partial t} + \frac{1}{2} \left[ \sigma_S^2 S^2 \frac{\partial^2 V}{\partial S^2} - 2\rho \sigma_S \sigma_Y SY \frac{\partial^2 V}{\partial S \partial Y} + \sigma_X^2 Y^2 \frac{\partial^2 V}{\partial Y^2} \right] \nonumber \\
& \qquad + r_pS \frac{\partial V}{\partial S} + (r_p-r_d)Y \frac{\partial V}{\partial Y} - r_pV = 0, \label{eq041}\\
& V(S, Y, T) = \max (S - K_dY, 0). \label{eq042}
\end{align}     
Here $K_d$ is constant. The maturity payoff function \eqref{eq042} satisfy the homogeneity condition and thus by the standard change of numeraire 
\begin{equation} \label{eq043}
U = \frac{V}{Y}, ~~ z=\frac{S}{Y}
\end{equation}   
(the essence of \eqref{eq043} is to change UK pound prices of option and stock to US dollar prices of them.) the problem \eqref{eq041} and \eqref{eq042} can be transformed into one dimensional problem: 
\begin{equation}\label{eq044}
\left\{
\begin{array}{l}
\frac{\partial U}{\partial t} + \frac{1}{2} (\sigma_S^2 + 2\rho \sigma_S \sigma_Y + \sigma_Y^2)z^2 \frac{\partial^2 U}{\partial z^2} + r_d z \frac{\partial U}{\partial z} -r_d U = 0, \\
U(z, T) = \max (z-K_d, 0).
\end{array}
\right.
\end{equation}    
If we consider \eqref{eq026}, then we can know that \eqref{eq044} is just the same with \eqref{eq034}. 
Thus using standard Black-Scholes formula and returning to the original variables we have 
\begin{equation} \label{eq045}
V_p(S, Y, t) = SN(d_1) - \frac{S_0}{Y_0}Ye^{-r_d(T-t)}N(d_2),
\end{equation}   
\begin{align*}
& d_1 = \frac{\ln \frac{SY_0}{YS_0} + \left( r_d+\frac{1}{2}\sigma_{S, Y}^2 \right)(T-t)}{\sigma_{S, Y}\sqrt{T-t}}, \\
& d_2 = d_1 - \sigma_{S, Y} \sqrt{T-t}.
\end{align*} 
Here $\sigma_{S, Y} = \sigma_{S, X}$ defined in \eqref{eq035}.
\begin{remark}
The formulae \eqref{eq037} and \eqref{eq045} are the same with those in \cite{ben}.
 \end{remark} 

%
%

\subsection{Savings Plans that provide a choice of indexing}

\textbf{A Practical Problem:} 
A bank in Israel has a service of savings account with choice item of interest rates, the holder of which has the right to select one among Israel interest rate and US interest rate. 
This saving contract is an {\it interest rate exchange} option. 
For example, if a saver today saves 100 NIS (unit of Israel currency) to NIS/USD savings account with maturity 1 year, then after 1 year the bank will pay the saver the maximum value among the following two value:
\begin{itemize}
\item Sum of 100 NIS and its interest of NIS (multiplied by the rate of inflation in Israel).
\item Sum of US dollar price of 100 NIS and its interest of USD (multiplied by NIS/USD exchange rate).  
\end{itemize}
Such savings plan is a kind of NIS/USD interest rate {\it exchange options} and it has a risk of NIS/USD exchange rate. What is the fair price of this saving account contract? \\

\noindent \textbf{Mathematical model:}
  
Let $r_d$ denote the domestic (risk free) interest rate,  $r_f$ the foreign (risk free) interest rate, $X(t)$ domestic currency / foreign currency exchange rate $(Y(t) = X(t)^{-1}$: foreign currency / domestic
currency exchange rate) and $I(t)$ the domestic inflation process (domestic {\it price level}). \\

\noindent \textbf{Assumptions:} 
\begin{enumerate}
\item	The invested quantity of money to the savings plan is 1 unit of (domestic) currency, $r_d$ and $r_f$ are constants. 
\item The maturity payoff is as follows:
\begin{align*} 
& V_d = \max \left\{ e^{r_dT}I(T), Y(0)e^{r_fT}X(T)\right\} ~ \textnormal{(domestic currency)}, \\
& V_f = V_d \cdot Y(T) = \max \left\{ e^{r_dT}I(T)Y(T), Y(0)e^{r_fT}\right\} ~ \textnormal{(foreign currency)}.
\end{align*}   
\item The exchange rate $Y(t)$ satisfies the model of Garman and Kohlhagen \cite{gar}:
\begin{equation*} 
dY(t) = Y(t)(r_f-r_d)dt + Y(t)\sigma_YdW^Y(t) ~ \textnormal{(risk neutral measure)}.
\end{equation*}  
Under this assumption from It\^o formula, for $X = Y^{-1}$ we have   
\begin{align*} 
dX & = -\frac{1}{Y^2}dY + \frac{1}{2}\frac{2}{Y^3}(dY)^2 \\
& = \left[ -\frac{1}{Y^2}Y(r_f-r_d) + \frac{1}{Y^3}Y^2\sigma_Y^2 \right]dt - \frac{1}{Y^2}Y\sigma_YdW^Y(t) \\
& = \frac{1}{Y}(r_d-r_f+\sigma_Y^2)dt - \frac{1}{Y}\sigma_YdW^Y(t).
\end{align*} 
Thus the exchange rate $Y(t)$ follows a geometric Brownian motion: 
\begin{equation*} 
dX(t) = \alpha_XX(t)dt + X(t) \sigma_X dW^X(t).
\end{equation*}   
Here
\begin{align*} 
& \sigma_X = \sigma_Y, \\
& dW^X(t) = -dW^Y(t) \\
& dW^I(t) \cdot dW^X(t) = -\rho dt, ~ \vert \rho \vert <1.
\end{align*}  
\item The domestic inflation process (price level) $I(t)$ satisfies (geometric Brownian motion)
\begin{equation*} 
dI(t) = \alpha_I I(t)dt + \sigma_I I(t) dW^I(t) ~ \textnormal{(risk neutral measure)}.
\end{equation*}   
Here $W^Y(t)$ and $W^I(t)$ are vector Wiener process with the following relation 
\begin{equation*} 
dW^I(t) \cdot dW^Y(t) = \rho dt, ~ \vert \rho \vert <1.
\end{equation*}  
\item The domestic price function of the option is a deterministic function $V_d = V(X, I, t)$.
\end{enumerate}

\noindent \textbf{PDE modeling of option price:} By $\Delta$-hedging, construct a portfolio $\Pi$
\begin{equation*} 
\Pi = V-\Delta_1X-\Delta_2I.
\end{equation*}  
(This portfolio consists of a sheet of option, $\Delta_1$ units of foreign currency and $\Delta_2$ units of domestic currency and its price is calculated in domestic currency.) 
We select $\Delta_1$ and $\Delta_2$ so that $\Pi$ becomes risk free in $(t, t + dt)$, that is,
\begin{equation*} 
d\Pi = r_d \Pi dt
\end{equation*} 
which is equivalent to
\begin{equation} \label{eq046}
dV-\Delta_1dX-\Delta_1r_fdt \cdot X - \Delta_2dI - \Delta_2 r_ddt \cdot I = r_d(V-\Delta_1X-\Delta_2I)dt.
\end{equation}  
By It\^o formula, we have
\begin{equation} \label{eq047}
dV = \left\{ \frac{\partial V}{\partial t}+\frac{1}{2}\left[ \sigma_X^2 X^2 \frac{\partial^2 V}{\partial X^2}-2\rho\sigma_X \sigma_I XI \frac{\partial^2 V}{\partial X \partial I} + \sigma_I^2 I^2 \frac{\partial^2 V}{\partial I^2} \right] \right\}dt + \frac{\partial V}{\partial X}dX + \frac{\partial V}{\partial I}dI.
\end{equation}     
Thus if we select
\begin{equation} \label{eq048}
\Delta_1 = \frac{\partial V}{\partial X}, ~~ \Delta_2 = \frac{\partial V}{\partial I},
\end{equation}  
and substitute \eqref{eq047} and \eqref{eq048} into \eqref{eq046}, then we have the PDE model of option price:
\begin{align}
& \frac{\partial V}{\partial t}+\frac{1}{2}\left[ \sigma_X^2 X^2 \frac{\partial^2 V}{\partial X^2}-2\rho\sigma_X \sigma_I XI \frac{\partial^2 V}{\partial X \partial I} + \sigma_I^2 I^2 \frac{\partial^2 V}{\partial I^2} \right] + (r_d-r_f)X \frac{\partial V}{\partial X} - r_dV = 0 \label{eq049} \\
& V(X, I, T) = \max \left\{ e^{r_dT}I(T), Y(0) e^{r_fT}X(T)\right\}. \label{eq050}
\end{align}   
The maturity payoff function \eqref{eq050} satisfy the homogeneity condition and thus the standard change of numeraire
\begin{equation} \label{eq051}
U=\frac{V}{X}, ~~ z=\frac{I}{X}
\end{equation}   
can reduce the special dimension. (The financial meaning of \eqref{eq051} is to change domestic prices of option and stock to foreign prices of them.)  If we consider
\begin{align*}
& \frac{\partial V}{\partial t} = \frac{\partial U}{\partial t}X, ~~ I \frac{\partial V}{\partial I} = z\frac{\partial U}{\partial z}X, ~~ X\frac{\partial V}{\partial X} = \left( -z \frac{\partial U}{\partial z}+U\right)X, \\
& I^2 \frac{\partial^2 V}{\partial I^2} = X^2 \frac{\partial^2 V}{\partial X^2} = z^2 \frac{\partial^2 U}{\partial z^2}Y, ~~ IX\frac{\partial^2 V}{\partial I \partial X} = -z^2 \frac{\partial^2 U}{\partial z^2} X, \\
& U(z, T) = e^{r_dT} \max \left\{ z, Y(0)e^{(r_f-r_d)T}\right\} \\
& \qquad \quad ~ = e^{r_dT}\left\{ Y(0)e^{(r_f-r_d)T} + \max \left[ z-Y(0)e^{(r_f-r_d)T}, 0\right] \right\},
\end{align*}   
the problem \eqref{eq049} and \eqref{eq050} can be transformed into one dimensional problem:
\begin{equation}\label{eq052}
\left\{
\begin{array}{l}
\frac{\partial U}{\partial t} + \frac{1}{2} (\sigma_X^2 + 2\rho \sigma_X \sigma_I + \sigma_I^2)z^2 \frac{\partial^2 U}{\partial z^2} + (r_f-r_d)z \frac{\partial U}{\partial z} -r_f U = 0, \\
U(z, T) = A+e^{r_dT} \max (z-K, 0).
\end{array}
\right.
\end{equation}  
Here
\begin{equation} \label{eq053}
A=Y(0)e^{r_fT}, ~~ K=Y(0)e^{(r_f-r_d)T}.  
\end{equation}  
Thus using standard Black-Scholes formula with risk free rate $r_f$, dividend rate $r_d$ and volatility $\sigma_{X, I}$ given by
\begin{equation} \label{eq054}
\sigma_{X, I}^2 = \sigma_{X}^2 + 2\rho \sigma_I \sigma_X + \sigma_I^2 > 0.
\end{equation}  
We have
\begin{equation*} 
U(z, t) = Ae^{-r_f(T-t)} + e^{-r_dT}\left( z e^{-r_d(T-t)}N(d_1) - Ke^{-r_f(T-t)}N(d_2)\right),
\end{equation*}   
\begin{align*}
& d_1 = \frac{\ln \frac{z}{K} + \left( r_f - r_d+\frac{1}{2}\sigma_{X, I}^2 \right)(T-t)}{\sigma_{X, I}\sqrt{T-t}}, \\
& d_2 = d_1 - \sigma_{X, I} \sqrt{T-t}.
\end{align*}  
Returning to original variables $(V, X, I)$ with consideration of \eqref{eq053}, we have
\begin{align*}
V_d(X, I, t) & = X(t)Y(0)e^{r_fT}e^{-r_f(T-t)} + e^{r_dT} I(t)e^{-r_d(T-t)}N(d_1) \\
& \qquad -X(t)Y(0)e^{(r_f-r_d)T}e^{r_dT}e^{-r_f(T-t)}N(d_2) \\
& = X(t)Y(0)e^{r_ft} + I(t)e^{r_dt}N(d_1) - X(t)Y(0)e^{r_ft}N(d_2) \\
& = I(t)e^{r_dt}N(d_1) + X(t)Y(0)e^{r_ft}(1-N(d_2)),
\end{align*} 
where $d_1$ can be simplified as follows:
\begin{align*}
& \ln \frac{I}{X \cdot Y(0)e^{(r_f-r_d)T}} + \left( r_f-r_d + \frac{1}{2}\sigma_{X, I}^2 \right) (T-t) \\
& \qquad = \ln \frac{I(t)}{X(t) \cdot Y(0)} - (r_f-r_d)T + (r_f-r_d)(T-t) + \frac{1}{2}\sigma_{X, I}^2 (T-t) \\
& \qquad = \ln \frac{I(t)e^{r_dt}}{X(t) \cdot Y(0)e^{r_ft}} + \frac{1}{2}\sigma_{X, I}^2 (T-t).
\end{align*} 
Therefore we obtain a representation of the \textit{\textbf{domestic price function}} of the option: 
\begin{equation} \label{eq055}
V_d(X, I, t) = I(t)e^{r_dt}N(\hat{d_1}) + X(t)Y(0) e^{r_ft}N(-\hat{d_2}),
\end{equation}   
\begin{align*}
& \hat{d}_1 = \frac{\ln \frac{I(t)e^{r_dt}}{X(t) \cdot Y(0)e^{r_ft}} + \frac{1}{2}\sigma_{X, I}^2 (T-t)}{\sigma_{X, I}\sqrt{T-t}}, \\
& \hat{d}_2 = \hat{d}_1 - \sigma_{X, I} \sqrt{T-t}.
\end{align*} 

The \textit{\textbf{financial meaning}} of this formula is clear: $I(t)e^{r_dt}$ is the price (calculated in domestic currency) at time $t$ of the money on the account where 1 unit of domestic currency is saved at time $t = 0$ after calculation of interest with domestic interest rate (under consideration of inflation rate). $X(t)Y(0)e^{r_ft}$ is the price (calculated in domestic currency) at time $t$ of the money on the account where 1 unit of domestic currency is saved at time $t = 0$ after calculation of interest with foreign interest rate. $N(\hat{d}_1)$  and $N(\hat{d}_2)$  show the portions of the two quantities $I(t)e^{r_dt}$  and $X(t)Y(0)e^{r_ft}$. The portions depend on which is greater and in particular, at maturity one is 0 and another one is 1.  
        
A representation of the \textit{\textbf{foreign (currency) price function}} of the option is given as follows:
\begin{equation} \label{eq056}
V_f = Y(t)V_d(t) = Y(t)I(t)e^{r_dt}N(\hat{d}_1) + Y(0) e^{r_ft}N(-\hat{d}_2).
\end{equation}    
Here $\hat{d}_1$  can be written by $Y(t)$ as follows:  
\begin{equation*}
\hat{d}_1 = \frac{\ln \frac{Y(t)I(t)e^{r_dt}}{Y(0)e^{r_ft}} + \frac{1}{2}\sigma_{X, I}^2 (T-t)}{\sigma_{X, I}\sqrt{T-t}}.
\end{equation*}  
In particular, the price of the option at $t = 0$ calculated by foreign currency is
\begin{equation}\label{eq057}
\begin{array}{l}
V_f(0) = Y(0)[I(0)N(\tilde{d}_1) + N(-\tilde{d}_2)] = Y(0)[1+I(0)N(\tilde{d}_1)-N(\tilde{d}_2)], \\
\tilde{d}_1 = \frac{\ln I(0)+\frac{1}{2}\sigma_{X, I}^2T}{\sigma_{X, I}\sqrt{T}}, ~~ \tilde{d}_2 = \tilde{d}_1 -\sigma_{X, I}\sqrt{T}.
\end{array}
\end{equation}   

\begin{remark}
In \cite{ben} they studied pricing problems of convertible bonds. Convertible bonds are a kind of derivatives of interest rate but most of models of interest rate (including  HJM, Vasicek, Cox-Ingersoll-Ross, Ho-Lee and Hull-White models) are not geometric Brownian motion. Thus most PDEs satisfied by prices of interest rate derivatives are not standard (multi-dimensional) Black-Scholes equations and our theorem 1 cannot be directly applied to reduce the dimension. However fortunately, the prices of risk free zero coupon bonds under Vasicek, Ho-Lee and Hull-White models follow geometric Brownian motion. Therefore, in what follows by seeing convertible bonds as a derivatives of zero coupon bond we derive multi-dimensional Black-Scholes type equations satisfied by convertible bonds and then apply the Theorem 1. 
\end{remark}

%
%

\subsection{Pricing Convertible Bonds}

\textbf{Problem:} A firm issues a zero coupon bond with maturity $T_1$ and face value 1. On a fixed day 
$T_0 (T_0 < T_1)$ the bond can be converted to a share of stock $S$. 
The stock $S$ does not pay dividend. 
Then we must find a representation of the convertible bond price at time $t < T_0$. \\

\noindent \textbf{Mathematical Model}

Let $S(t)$ be the price of stock at time $t, p = p(t, T)$ time $t$-price of risk free zero coupon bond
and $r(t)$ short rate. \\

\noindent \textbf{Assumptions}

\begin{enumerate}
\item	 The short rate follows Vacisek model (Ho-Lee model and Hull-White model can be considered with the same method): \label{ee2}
\begin{equation}\label{eq058}
dr_t = \theta (\mu_r-r_t)dt + \sigma_r dW(t).
\end{equation}   
Here the volatility $\sigma_r$ of short rate is a deterministic constant.
\item	 The stock price $S(t)$ follows the following geometric Brownian motion:
\begin{equation*}
dS(t) = \alpha_S S(t)dt + \sigma_SS(t)dW^S(t).
\end{equation*}  
Here the volatility $\sigma_S$ of stock price is a deterministic constant. 
\item	 $W(t)$ and $W^S(t)$ are vector Wiener processes with the following relations:
\begin{equation*}
dW(t) \cdot dW^S(t) = \rho dt, \vert \rho \vert <1.
\end{equation*}  
\item	 The price function of the convertible bond is given by a deterministic function $V(S, r, t)$. \label{ee3}
\end{enumerate}

\noindent \textbf{The price of risk free zero coupon bond:} Under the assumption \ref{ee2} the price $p(r, t; T)$ of zero coupon bond with maturity $T$ and face value 1 satisfies the following problem \cite{wil}:
\begin{equation} \label{eq059}
\left\{
\begin{array}{l}
\frac{\partial p}{\partial t} + \frac{1}{2} \sigma_r^2 \frac{\partial^2 p}{\partial r^2}-rp = (\lambda \sigma_r - \theta(\mu_r - r)) \frac{\partial p}{\partial r}, \\
p(r, T; T) = 1.
\end{array}
\right.
\end{equation}   
Here $\lambda$ is market risk price. The solution to \eqref{eq059} is given as  
\begin{equation*}
p(r, t; T) = A(t)e^{-B(t)r}, ~~ \frac{\partial p}{\partial r} = -B(t)p, ~~ B(t) = \frac{1-e^{-\theta(T-t)}}{\theta},
\end{equation*}   
from which we can know that the short rate $r(t)$ is a deterministic function of zero coupon bond price $p = p(t, T)$:
\begin{equation} \label{eq060}
r = r(t, p) = -\frac{1}{B(t)}(\ln p-A(t)) = - \frac{1}{B(t)}\ln p + \frac{A(t)}{B(t)}.
\end{equation}  
By It\^o formula, \eqref{eq058} and \eqref{eq059} we have 
\begin{align*}
dp & = \left( \frac{\partial p}{\partial t} + \frac{1}{2} \sigma_r^2 \frac{\partial^2 p}{\partial r^2} \right)dt + \frac{\partial p}{\partial r}dr \\
& = \left( \frac{\partial p}{\partial t} + \frac{1}{2} \sigma_r^2 \frac{\partial^2 p}{\partial r^2} + \theta(\mu_r-r) \frac{\partial p}{\partial r}\right)dt + \sigma_r \frac{\partial p}{\partial r}dW^r(t) \\
& = \left( rp + \lambda\sigma_r \frac{\partial p}{\partial r} \right)dt +\sigma_r \frac{\partial p}{\partial r} dW^r(t) \\
& = (r-\lambda \sigma_rB(t))pdt - \sigma_r B(t) pdW^r(t). 
\end{align*}               
Thus zero coupon bond price process follows a geometric Brownian motion:
\begin{equation*}
dp = (r - \lambda \sigma_rB(t))pdt - \sigma_rB(t)pdW^r(t).
\end{equation*}   

\noindent \textbf{PDE pricing model:} Let denote
\begin{equation*}
p = p(t, T_1), ~~ \Sigma_p = \Sigma_p(t, T_1) = \sigma_r B(t, T_1). 
\end{equation*}  
The maturity $T_0$-payoff of our convertible bond is presented by
\begin{equation*}
V(T_0) = \max [S(T_0), p(T_0, T_1)].
\end{equation*}   
From \eqref{eq060} and assumption \ref{ee3} we can know that the price function of the convertible bond is given by a deterministic function $V(S, p, t)$. 

In order to derive the equation satisfied by $V(S, p, t)$, by $\Delta$-hedging, we construct a portfolio $\Pi$
\begin{equation*}
\Pi = V-\Delta_1S-\Delta_2p.
\end{equation*}   
(This portfolio consists of a sheet of bond, $\Delta_1$ shares of stock and $\Delta_2$ sheets of risk free zero coupon bonds.) We select $\Delta_1$ and $\Delta_2$ so that $\Pi$ becomes risk free in $(t, t + dt)$, that is,
\begin{equation*}
d\Pi = r(t, p) \Pi dt
\end{equation*}
which is equivalent to
\begin{equation} \label{eq061}
dV - \Delta_1dS - \Delta_2dp = r(t, p)\cdot (V-\Delta_1S-\Delta_2p)dt.
\end{equation} 
By It\^o formula, we have
\begin{equation} \label{eq062}
dV = \left\{ \frac{\partial V}{\partial t}+\frac{1}{2}\left[ \sigma_S^2 S^2 \frac{\partial^2 V}{\partial S^2}-2\rho\sigma_S \Sigma_p Sp \frac{\partial^2 V}{\partial S \partial p} + \Sigma_p^2 p^2 \frac{\partial^2 V}{\partial p^2} \right] \right\}dt + \frac{\partial V}{\partial S}dS + \frac{\partial V}{\partial p}dp.
\end{equation}      
Thus if we select
\begin{equation*}
\Delta_1 = \frac{\partial V}{\partial S}, ~~ \Delta_2 = \frac{\partial V}{\partial p},
\end{equation*} 
and substitute \eqref{eq062} into \eqref{eq061}, then we have the PDE model of convertible bonds:
\begin{align}
& \frac{\partial V}{\partial t} + \frac{1}{2} \left[ \sigma_S^2 S^2 \frac{\partial^2 V}{\partial S^2} - 2\rho \sigma_S \Sigma_p Sp \frac{\partial^2 V}{\partial S \partial p} + \Sigma_p^2 p^2 \frac{\partial^2 V}{\partial p^2} \right] \nonumber \\
& \qquad \qquad \qquad + r(t, p)S \frac{\partial V}{\partial S} + r(t, p)p \frac{\partial V}{\partial p} - r(t, p)V = 0, \label{eq063}\\
& V(S, p, T_0) = \max (S, p). \label{eq064}
\end{align}    
  
The equation \eqref{eq063} seems a standard two dimensional Black-Scholes equation but the coefficients of second order derivatives depend on time and furthermore the coefficients of first order derivatives {\it depend on time and special variables}. Thus strictly speaking, this example is not contained in the range of application of theorem 1 but the maturity payoff \eqref{eq064} has homogeneity and the change of numeraire 
\begin{equation} \label{eq065}
U = \frac{V}{p}, ~~ z = \frac{S}{p}
\end{equation}  
does work well. (The financial meaning of \eqref{eq065} is to change prices of bond and stock to {\it relative} prices on risk free zero coupon bond price.)  Then we have one dimensional problem:
\begin{equation} \label{eq066}
\left\{
\begin{array}{l}
\frac{\partial U}{\partial t} + \frac{1}{2} \left( \sigma_S^2 + 2\rho\sigma_S \Sigma_p (t, T_1) + \Sigma_p^2(t, T_1) \right) z^2 \frac{\partial^2 U}{\partial z^2} = 0, \\
U(z, T_0) = 1+ \max (z-1, 0).
\end{array}
\right.
\end{equation}   
Here note that the coefficient depends on time: 
\begin{equation*}
\sigma_z(t) = \sqrt{\sigma_S^2 + 2\rho\sigma_S\Sigma_p(t, T_1) + \Sigma_p^2(t, T_1)}.
\end{equation*} 
Using a standard method of Jiang (\cite{jia}, Chap. 5, (5.4.10)) we get the solution to \eqref{eq066}:
\begin{equation}\label{eq067}
\begin{array}{l}
U(z, t) = 1+zN(\bar{d}_1) - N(\bar{d}_2), \\
\bar{d}_1 = \frac{\ln z+\frac{1}{2}\sigma^2(t, T_0)}{\sigma(t, T_0)}, ~~ \bar{d}_2 = \bar{d}_1 -\sigma(t, T_0).
\end{array}
\end{equation}    
Here
\begin{align}
\sigma^2(t, T_0) = \int_t^{T_0} \sigma_z^2(u)du & = \int_t^{T_0} \left[ \sigma_S^2 + 2\rho\sigma_S\Sigma_p(u, T_1) + \Sigma_p^2(u, T_1) \right] du   \nonumber \\
& = \int_t^{T_0} \left[ \sigma_S^2 + 2\rho\sigma_S\sigma_r B(u, T_1) + \sigma_r^2B^2(u, T_1) \right] du. \label{eq068}
\end{align}
Returning to original variables $V, S$ and $p(t, T_1)$, we get the price of convertible bond:
\begin{equation}\label{eq069}
V(S, p, t) = p(t, T_1) + SN(\tilde{d}_1) - p(t, T_1)N(\tilde{d}_2),
\end{equation}
\begin{equation}\label{eq070}
\tilde{d}_1 = \frac{\ln \frac{S}{p(t, T_1)} + \frac{1}{2}\sigma^2(t, T_0)}{\sigma^2(t, T_0)}, ~~ \tilde{d}_2 = \tilde{d}_1 - \sigma^2(t, T_0).
\end{equation}

\begin{remark} The formula \eqref{eq069} is the same with the one in \cite{ben} where they used HJM model of short rate while we used Vasicek model. 
That is why zero coupon bond price becomes a deterministic function of short rate under Vasicek model. 
\end{remark}


\subsection{Pricing Corporate Bonds which can be converted to stock at maturity}

\textbf{Problem:} The holder of this contract has the right to convert the bond to $\alpha$ shares of stock at maturity. The price of such bond is related to not the price of a sheet of bond but the firm value. Note that firm value is not tradable. What is the fair price of this contract? 

This problem was studied in Li and Ren \cite{li}. Their equation for pricing is not a standard Black-Scholes equation. \\

\noindent \textbf{Mathematical Modeling} 

\noindent \textbf{Assumptions}
\begin{enumerate}
\item Let assume that a firm issued m shares of stock with time $t$-stock price $S(t)$ per share and $n$ sheets of bond which has with time $t$-price $C_t$. Then firm value $V(t)$ is
\begin{equation}\label{eq071}
V_t = mS_t + nC_t
\end{equation} 
\item Risk free short rate $r_t$ follows Vasicek model:
\begin{equation}\label{eq072}
dr_t = \theta(\mu_r - r_t)dt + \sigma_rdW^r(t).
\end{equation} 
\item Firm value $V(t)$ follows a geometric Brownian motion
\begin{equation*}
dV(t) = \mu_V V(t)dt + \sigma_V V(t)dW^V(t)
\end{equation*}  
where drift $\mu_V$ and $\sigma_V$ are all constants.
\item $W_r(t)$ and $W_V(t)$ are vector Wiener processes with the following relation
\begin{equation*}
dW^r(t) \cdot dW^V(t) = \rho dt, ~~ \vert \rho \vert < 1.
\end{equation*}  
\item The bond price is given by a deterministic function $C(V, r, t)$. \label{ee4}
\item The maturity payoff is given by $C_T = \max \{K, \alpha S_T\}$.  Here $K$ is maturity time $T$-payoff (the sum of principal and coupon at time $T$) and $\alpha$ is convertible rate.

By \eqref{eq071}, if the bond is converted, then we have
\begin{equation*}
V_T = mS_T + n\alpha S_T.
\end{equation*}
Therefore we have 
\begin{equation}\label{eq073}
C_T = \max \left\{ K, \frac{\alpha}{m+n\alpha}V_T \right\}.
\end{equation} 
\end{enumerate} 

\noindent \textbf{Zero coupon bond price process:} Let $p(t, T)$ be the price of zero coupon bond with maturity $T$. 
Then as described in subsection 3.4, short rate $r(t)$ is a deterministic function of zero coupon bond price $p = p(t, T)$ and furthermore zero coupon bond price process follows the following geometric Brownian motion: 
\begin{equation}\label{eq074}
dp = \alpha(t)pdt - \Sigma_p(t)pdW^r(t),
\end{equation}  
where $\Sigma_p(t)$ is a function of volatility $\sigma_r$ of short rate $r$ and time $t$ as in the subsection 3.4. \\

\noindent \textbf{PDE modeling:} Since $p(T, T) = 1$, \eqref{eq073} can be rewritten as
\begin{equation}\label{eq075}
C_r = \max \left\{ Kp, \frac{\alpha}{m+n\alpha}V_T\right\}.
\end{equation}   
From the assumption \ref{ee4} and $r(t) = r(t, p_t)$ we can know that the price function of the convertible bond is given by a deterministic function $C(V, p, t)$. 

In order to derive the equation satisfied by $C(V, p, t)$, by $\Delta$-hedging, we construct a portfolio $\Pi$ 
\begin{equation*}
\Pi = V - \Delta_1S - \Delta_2p.
\end{equation*}
(This portfolio consists of a sheet of convertible bond, $\Delta_1$ shares of stock and $\Delta_2$ sheets of risk free zero coupon bonds. Note that firm value is not tradable.) 
We select $\Delta_1$ and $\Delta_2$ so that $\Pi$ becomes risk free in $(t, t + dt)$, that is,
\begin{equation*}
d\Pi = r(t, p) \Pi dt
\end{equation*}
which is equivalent to
\begin{equation}\label{eq076}
d\Pi = dC - \Delta_1dS - \Delta_2dp = r(t, p)\cdot (C-\Delta_1S-\Delta_2p)dt.
\end{equation}    
By It\^o formula, we have
\begin{equation} \label{eq077}
dC = \left\{ \frac{\partial C}{\partial t}+\frac{1}{2}\left[ \sigma_V^2 V^2 \frac{\partial^2 C}{\partial V^2}-2\rho\sigma_V \Sigma_p Vp \frac{\partial^2 C}{\partial V \partial p} + \Sigma_p^2 p^2 \frac{\partial^2 C}{\partial p^2} \right] \right\}dt + \frac{\partial C}{\partial V}dV + \frac{\partial C}{\partial p}dp.
\end{equation}       
From \eqref{eq071} we have
\begin{equation}\label{eq078}
dS = \frac{dV-ndC}{m},
\end{equation}    
If we substitute \eqref{eq077}, \eqref{eq078} and \eqref{eq071} into \eqref{eq076}, then we have
\begin{align*}
d\Pi & = dC - \frac{\Delta_1}{m}dV + \frac{\Delta_1n}{m}dC - \Delta_2dp \\
& = \left( 1+ \frac{\Delta_1n}{m} \right)dC - \frac{\Delta_1}{m}dV - \Delta_2dp \\
& = \left( 1+ \frac{\Delta_1n}{m} \right) \left\{ \left[\frac{\partial C}{\partial t}+\frac{1}{2}\left[ \sigma_V^2 V^2 \frac{\partial^2 C}{\partial V^2}-2\rho\sigma_V \Sigma_p Vp \frac{\partial^2 C}{\partial V \partial p} + \Sigma_p^2 p^2 \frac{\partial^2 C}{\partial p^2} \right] \right]dt \right. \\
& \qquad \qquad \left. + \frac{\partial C}{\partial V}dV + \frac{\partial C}{\partial p}dp \right\} - \frac{\Delta_1}{m}dV - \Delta_2dp 
\end{align*} 
\begin{align*}
& \quad =  \left( 1+ \frac{\Delta_1n}{m} \right)\left[\frac{\partial C}{\partial t}+\frac{1}{2}\left[ \sigma_V^2 V^2 \frac{\partial^2 C}{\partial V^2}-2\rho\sigma_V \Sigma_p Vp \frac{\partial^2 C}{\partial V \partial p} + \Sigma_p^2 p^2 \frac{\partial^2 C}{\partial p^2} \right] \right]dt \\
& \quad \qquad \qquad + \left[ \left( 1+ \frac{\Delta_1n}{m} \right)\frac{\partial C}{\partial V} - \frac{\Delta_1}{m} \right] dV + \left[ \left( 1+ \frac{\Delta_1n}{m} \right)\frac{\partial C}{\partial p} - \Delta_2 \right] dp \\
& \quad = r \left[ \left( 1+ \frac{\Delta_1n}{m} \right)C - \frac{\Delta_1}{m}V - \Delta_2p \right]dt.
\end{align*}                  
Here if we select so that 
\begin{equation*}
\left( 1+ \frac{\Delta_1n}{m} \right) \frac{\partial C}{\partial V} - \frac{\Delta_1}{m} = 0, ~~ \left( 1+ \frac{\Delta_1n}{m} \right) \frac{\partial C}{\partial p} - \Delta_2 = 0,
\end{equation*} 
that is,
\begin{equation} \label{eq079}
\Delta_1 = \frac{m\frac{\partial C}{\partial V}}{1-n \frac{\partial C}{\partial V}}, ~~  \Delta_2 = \frac{\frac{\partial C}{\partial p}}{1-n \frac{\partial C}{\partial V}},
\end{equation}  
and multiply $1-n \cdot \frac{\partial C}{\partial V}$  to the two sides, then we have 
\begin{align} \label{eq080}
\frac{\partial C}{\partial t} &+\frac{1}{2}\left[ \sigma_V^2 V^2 \frac{\partial^2 C}{\partial V^2}-2\rho\sigma_V \Sigma_p Vp \frac{\partial^2 C}{\partial V \partial p} + \Sigma_p^2 p^2 \frac{\partial^2 C}{\partial p^2} \right] \nonumber \\
& \qquad + r(t, p)V \frac{\partial C}{\partial V} + r(t, p)p \frac{\partial C}{\partial p} - r(t, p)C = 0.
\end{align}  
So we obtain the pricing model \eqref{eq080} and \eqref{eq075} of the convertible bond. 
The equation \eqref{eq080} is just the same with \eqref{eq063} and the maturity payoff \eqref{eq075} has homogeneity. 
Thus change of numeraire 
\begin{equation} \label{eq081}
U = \frac{C}{p}, ~~ z = \frac{V}{p}.
\end{equation} 
transforms \eqref{eq080} and \eqref{eq075} into the following one dimensional Black-Scholes terminal value problem:
\begin{equation} \label{eq082}
\left\{
\begin{array}{l}
\frac{\partial U}{\partial t} + \frac{1}{2} \left( \sigma_V^2 + 2\rho\sigma_V \Sigma_p (t) + \Sigma_p^2(t) \right) z^2 \frac{\partial^2 U}{\partial z^2} = 0, \\
U(z, T_0) = K+ \frac{\alpha}{m+n\alpha} \left[ z-K\left(n+\frac{m}{\alpha}\right)\right]^+.
\end{array}
\right.
\end{equation}    
Using a standard method we get the solution to \eqref{eq082}:
\begin{equation*}
U(z, t) = K+\frac{\alpha}{m+n\alpha}zN(\bar{d}_1) - KN(\bar{d}_2), 
\end{equation*}    
\begin{equation*}
\bar{d}_1 = \frac{\ln \frac{\alpha z}{K(m+n\alpha)} + \frac{1}{2}\sigma^2(t, T)}{\sigma(t, T)}, ~~ \bar{d}_2 = \bar{d}_1 -\sigma(t, T).
\end{equation*}    
Here
\begin{equation} \label{eq083}
\sigma^2(t, T) = \int_t^T \left[ \sigma_V^2 + 2\rho \sigma_V \Sigma_p(u) + \Sigma_p^2(u) \right]du.
\end{equation}   
Returning to original variables $C, V$ and $p(t, T_1)$, we get the price of corporate bond with convertible clause at maturity:
\begin{equation}\label{eq084}
C(S, p, t) = Kp(t, T) + \frac{\alpha}{m+n\alpha}VN(\tilde{d}_1) - Kp(t, T)N(\tilde{d}_2),
\end{equation}
\begin{align*}
& \tilde{d}_1 = \frac{\ln \left( \frac{V}{Kp(t, T)} \cdot \frac{\alpha}{(m+n\alpha)} \right) + \frac{1}{2}\sigma^2(t, T)}{\sigma(t, T)}, \\
& \tilde{d}_2 = \tilde{d}_1 - \sigma^2(t, T).
\end{align*}

\begin{remark} 
 The formula \eqref{eq084} is the same with the one in Li and Ren \cite{li}.
\end{remark} 

%
%

\section{Conclusions}

Geman \cite{gem} and Jamshidian \cite{jam} gave the mathematical foundation of change numeraire, and since then the change 
of numeraire has been an important tool for complicated option pricing problems to simplify \cite{ben,jia,li}. 
In \cite{ben}, they explained the essence of numeraire approach in the viewpoint of stochastic calculus. 
In this paper we tried to understand the essence of numeraire approach in the viewpoint of PDE theory with 5 concrete 
practical problems in finance as examples. 
The numeraire approach is not an all-purpose tool for complicated option pricing problems but through this study we can know the {\bf following two facts}: 
  
1) When there are several risk sources and maturity payoff function is a homogeneous function on risk source variables, 
we can reduce the dimension of pricing problem of financial derivative by 1 using one risk source variable as a numeraire.

2) The reason that numeraire approach works well in pricing problem of a financial derivative is that the pricing problem itself, 
in fact, has a structure of Black-Scholes equation (see the subsections 3.4 and 3.5) and Black-Scholes equations have a  special preserving property under change of numeraire.


\begin{thebibliography}{99}


\bibitem{ben}
Benninga, S., Bj\"ork, T.  and Wiener, Z. On  the  use  of  numeraires  in option  pricing. \textit{The Journal of Derivatives}. Winter 2002. \textbf{10}(2): 1-16.

\bibitem{bjo}
Bj\"ork, T. \textit{Arbitrage theory in continuous time}. Oxford university press, 1999.

\bibitem{bre}
Brenner, M. and Galai, D. The determinants of the return on index bonds. \textit{Journal of Banking} \& \textit{finance}. 1978. \textbf{2}(1): 47-64. 

\bibitem{fis}
Fisher, S. Call option pricing when the exercise price is uncertain, and the valuation of index bonds. \textit{The Journal of Finance}. 1978. \textbf{33}(1): 169-176. 

\bibitem{gar}
Garman, M. and Kohlhagen, S. Foreign  currency  option  values. \textit{Journal  of International Money and Finance}. 1983. \textbf{2}(3): 231-237.

\bibitem{gem}
Geman, H. The importance of forward neutral probability in a stochastic approach of interest rates. \textit{Working paper}. ESSEC: 1989. 

\bibitem{har}
Harrison, J. and Kreps, J. Martingales and arbitrage in multiperiod securities markets. \textit{Journal of Economic Theory}. 1979. \textbf{20}(3): 381-408.

\bibitem{jam}
Jamshidian, F. An exact bond option formula. \textit{The Journal of Finance}. 1989. \textbf{44}(1): 205-209.

\bibitem{jia}
Jiang, L. \textit{Mathematical modeling and methods of option pricing}. Beijing: Higher education press. 2003. chap. 7 (in  Chinese),  English  translation:  Jiang, Li-shang. \textit{Mathematical modeling and methods of option pricing}. Singapore: World Scientific. 2005.

\bibitem{li}
Li, S. and Ren, X. Pricing Theory of Convertible bond (1), \textit{System Engineering; Theory and Practice}. 2004. \textbf{24}(8): 18-25 (in Chinese). 

\bibitem{mar}
Margrabe, W. The value of an option to exchange one asset for another. \textit{The Journal of Finance}. 1978. \textbf{33}(1): 177-186. 

\bibitem{mer}
Merton, R. C. Theory of rational option pricing. \textit{The Bell Journal of Economics and Management Science}. 1973. \textbf{4}(1): 141-183. 

\bibitem{wil}
Wilmott, P. \textit{Derivatives: The  Theory  and  Practice  of  Financial  Engineering}. New York: John Wiley \& Sons, Inc.1998.  


\end{thebibliography}
\end{document}